\documentclass[11pt]{amsart} 

\usepackage[utf8]{inputenc}
\usepackage[T1]{fontenc}	
\usepackage[french,english]{babel}	

\usepackage{lmodern}			

\usepackage{graphicx}	

\usepackage[a4paper,plainpages=false,colorlinks,linkcolor=bleuFonce,citecolor=rougeFonce,urlcolor=vertFonce,breaklinks]{hyperref}

\usepackage{placeins}
\usepackage{float} 

\usepackage{tikz}
\usetikzlibrary{patterns}

\usepackage[normalem]{ulem}
\usepackage{enumitem}

\usepackage{color}
\definecolor{vertFonce}	{rgb}{0,0.5,0}
\definecolor{numLignes}	{rgb}{0.17,0.57,0.7}	
\definecolor{gris}		{rgb}{0.5,0.5,0.5}
\definecolor{grisFonce}	{rgb}{0.2,0.2,0.2}
\definecolor{orange}	{rgb}{1,0.65,0.31}		
\definecolor{orangeFonce}{rgb}{1,0.4,0}
\definecolor{bleuFonce}	{rgb}{0,0,0.4}
\definecolor{rougeFonce}{rgb}{0.3,0,0}
\definecolor{rougeWord}	{rgb}{0.5,0,0}
\definecolor{vertClair}	{rgb}{0.8,1,0.8}
\definecolor{rougeClair}{rgb}{1,0.5,0.5}
\definecolor{violet}	{rgb}{0.5,0,0.5}

\usepackage{pict2e}
\usepackage{multido}

\usepackage{amsfonts,amssymb,amsthm,amsmath} 
\usepackage{dsfont}				
\usepackage{mathrsfs}
\usepackage[mathscr]{euscript}
\usepackage{yfonts}
\usepackage{cancel}
 
\newtheorem{lem}{Lemma}[section]
\newtheorem{theorem}{Theorem}[section]

\newtheorem{prop}{Proposition}[section]

\newtheorem{remark}{Remark}[section]

\usepackage{stmaryrd}
\SetSymbolFont{stmry}{bold}{U}{stmry}{m}{n}
\newcommand		{\subsetArrow}	{\mathrel{\ooalign{$\subset$\cr%
\hidewidth\raise-.087ex\hbox{$_\shortrightarrow\mkern-1.5mu$}\cr}}}
\newcommand		{\subsetarrow}	{\mathrel{\ooalign{$\subset$\cr%
\hidewidth\raise-1.45ex\hbox{$\vec{}\mkern6mu$}\cr}}}

\newcommand		{\N}		{\mathbb N}			
\newcommand		{\RR}		{\mathbb R}			
\newcommand		{\R}		{\RR}
\newcommand		{\Z}		{\mathbb Z}
\newcommand		{\Rd}		{\R^d}
\newcommand		{\RRd}		{\R^{2d}}
\newcommand		{\CC}		{\mathbb C}			
\newcommand		{\cN}		{\mathcal N}		
\renewcommand	{\L}		{\mathcal L}		
\newcommand		{\cW}		{\mathcal W}		

\newcommand		{\cB}		{\mathcal B}

\newcommand		{\cX}		{\mathcal X}
\newcommand		{\cZ}		{\mathcal Z}

\newcommand		\sfA		{\mathsf A}
\newcommand		\sfB		{\mathsf B}

\newcommand		\sfH		{\mathsf H}			

\newcommand     {\ecF}      {\mathscr{F}}

\newcommand     {\ecS}      {\mathscr{S}}

\newcommand		{\lt}			{\left}				%
\newcommand		{\rt}			{\right}			%
\renewcommand	{\(}			{\lt(}
\renewcommand	{\)}			{\rt)}
\newcommand		{\bangle}[1]	{\lt\langle #1\rt\rangle}

\newcommand		{\inprod}[2]	{\bangle{#1, #2}}
\newcommand		{\com}[1]		{\lt[{#1}\rt]}		

\newcommand		{\n}[1]			{\lt\lvert #1 \rt\rvert}
		
\newcommand		{\nrm}[1]		{\lt\lVert #1\rt\rVert}
\newcommand		{\Nrm}[2]		{\nrm{#1}_{#2}}

\newcommand		{\floor}[1]		{\lt\lfloor{#1}\rt\rfloor}

\newcommand		{\indic}	{\mathds{1}}		

\renewcommand		{\d}		{\mathrm{d}}		
\newcommand			{\dd}		{\,\d}				
			
\newcommand			{\dpt}		{\partial_t}

\newcommand			{\dt}		{\frac{\d}{\d t}}	

\newcommand			{\grad}		{\nabla}
\newcommand			{\lapl}		{\Delta}

\newcommand			{\conj}[1]	{\overline{#1}}		

\DeclareMathOperator{\cF}		{\mathcal{F}}		
\DeclareMathOperator{\re}		{Re}				
\DeclareMathOperator{\tr}		{Tr}				

\newcommand			{\id}		{\mathbf{1}}

\newcommand			{\gibbs}	{\opom_N} 

\newcommand			{\sh}		{\operatorname{sh}}
\renewcommand		{\th}		{\operatorname{th}}

\renewcommand	{\Re}[1]		{\re\!\( #1 \)}		
\newcommand		{\Tr}[1]		{\tr\!\( #1 \)}		

\newcommand		\sfm		{\mathsf m}

\newcommand		{\intd}			{\int_{\Rd}}

\newcommand		{\iintd}		{\iint_{\RRd}}

\newcommand		{\init}			{\mathrm{in}}

\newcommand		{\cC}			{\mathcal{C}}
\newcommand		{\kb}			{k_{\mathrm{B}}}

\usepackage{braket}

\newcommand		{\op}		{\boldsymbol{\rho}}	

\newcommand		{\opal}		{\boldsymbol{\alpha}}
\newcommand		{\opG}		{\boldsymbol{\Gamma}}

\newcommand		{\opz}		{\boldsymbol{z}}	

\newcommand		{\opom}		{\boldsymbol{\omega}}

\newcommand		{\opg}		{\boldsymbol{g}}
\newcommand		{\opp}		{\boldsymbol{p}}

\newcommand		{\Dh}		{\boldsymbol{\nabla}}	
\newcommand		{\Dhx}[1]	{\Dh_{\!x} #1}			
\newcommand		{\Dhv}[1]	{\Dh_{\!\xi} #1}		

\DeclareMathOperator{\dG}	{\d\Gamma}			

\newcommand		{\h}		{\mathfrak{h}}		

\title[On the semiclassical regularity of thermal equilibria]{On the semiclassical regularity of thermal equilibria}
\author[J. Chong]{Jacky J. Chong}
\author[L. Lafleche]{Laurent Lafleche}
\address[J. Chong, L. Lafleche]{Department of Mathematics, The University of Texas at Austin, Austin, TX 78712, USA, {\tt jwchong@math.utexas.edu}, {\tt lafleche@math.utexas.edu}}

\author[C. Saffirio]{Chiara Saffirio}
\address[C. Saffirio]{Department of Mathematics and Computer Science, University of Basel, 4051 Basel, Switzerland, {\tt chiara.saffirio@unibas.ch}}


\begin{document}

\begin{abstract}
	We study the regularity properties of fermionic equilibrium states at finite positive temperature and show that they satisfy certain semiclassical bounds. As a corollary, we identify explicitly a class of positive temperature states satisfying the regularity assumptions of [J.J.~Chong, L.~Lafleche, C.~Saffirio: arXiv:2103.10946 (2021)].
\end{abstract}

\maketitle

\bigskip

\renewcommand{\contentsname}{\centerline{Table of Contents}}
\setcounter{tocdepth}{2}	
\tableofcontents


\section{Introduction}

	We consider a system of $N$ non-interacting fermions in a harmonic trap and study the semiclassical structure of its one-particle density matrix associated with the corresponding Gibbs (equilibrium) state of the trapped gas at temperature $T>0$. More precisely, given a family of one-particle normalized nonnegative trace class operators $\op_{\beta}$ acting on $L^2(\R^d)$ and indexed by the parameter $\beta$, proportional to $T^{-1}$, we consider the following quantities
	\begin{equation}\label{eq:quantum-grad}
		\Dhx\op_{\beta} :=\com{\nabla,\op_{\beta}} \quad\text{ and }\quad \Dhv\op_{\beta} :=\com{\tfrac{x}{i\hbar},\op_{\beta}},
	\end{equation}
	which we refer to as the quantum gradients of $\op_{\beta}$, 
	and show that for any $p\in [2,\infty]$,
	\begin{equation*}
		\op_\beta,\, \sqrt{\op_\beta} \in \cW^{1,p}(\sfm)
	\end{equation*}
	uniformly in $\hbar$, where $\cW^{1,p}(\sfm)$ are the semiclassical analogue of weighted Sobolev spaces equipped with the norm 
	\begin{equation*}
		\Nrm{\op}{\cW^{1,p}(\sfm)}^p := \Nrm{\op\,\sfm}{\L^p}^p+\Nrm{\Dhx\op\,\sfm}{\L^p}^p + \Nrm{\Dhv\op\,\sfm}{\L^p}^p
	\end{equation*}
	 and $\sfm:=1+\n{\opp}^n$. Here $\opp=-i\hbar\nabla$ is the momentum operator and $\L^p$ are the semiclassical versions of the Lebesgue spaces endowed with the rescaled Schatten norms
	 \begin{equation*}
	 	\Nrm{\op}{\L^p} := h^{\frac{3}{p}} \Tr{\n{\op}^p}^{\frac{1}{p}}
	 \end{equation*}
	 where $h = 2\pi\hbar$.
	
	The quantities~\eqref{eq:quantum-grad} play key roles in obtaining semiclassical estimates, such as in the problems of the derivation of the Vlasov equation from quantum mechanics and of the Hartree and the Hartree--Fock equations from many-body quantum mechanics. To the author's knowledge, they first appeared in~\cite{benedikter_mean-field_2014}, where the Hartree equation was derived in the mean-field regime for arbitrary long times from a system of $N$ fermions interacting through a smooth potential, under the assumption that the initial state $\op^\init$ of the Hartree dynamics is a pure state, i.e. $(\op^\init)^2=\op^\init$, and satisfies the following semiclassical structures
	\begin{equation}\label{eq:semiclassical-bounds}
		\Nrm{\Dhx\op^\init}{\L^1} \leq C \quad\text{ and }\quad \Nrm{\Dhv\op^\init}{\L^1}\leq C.
	\end{equation}
	Similarly, bounds on semiclassical Schatten norms of commutators are used in~\cite{golse_empirical_2019} to obtain uniform in $\hbar$ mean-field estimates, in~\cite{lafleche_strong_2021} in the context of the semiclassical limit, in \cite{golse_convergence_2021} in the context of the convergence of numerical schemes, in~\cite{golse_semiclassical_2020} to get estimates on the quantum Wasserstein distance.
	
	When the interaction is the Coulomb potential only partial results are available on a time scale of order one (see \cite{porta_mean_2017}) and they rely on a semiclassical structure expressed in terms of $L^p$ norms of the diagonal of the operator $|\Dhv\op|$, with $p>3$ (see also \cite{lafleche_strong_2021}). 
	
	In~\cite{fournais_optimal_2019} and~\cite{benedikter_effective_2022}, the authors obtained semiclassical bounds of the form~\eqref{eq:semiclassical-bounds} for pure states in the non-interacting case, which in turn provided explicit examples of initial pure states that satisfy the assumptions in~\cite{benedikter_mean-field_2014} on the semiclassical structures. In this paper, we extend this analysis to the finite positive temperature case, where the one-particle reduced density matrix $\op$ is no longer a projection. When dealing with mixed states the analogue relevant quantities considered in~\cite{benedikter_mean-field_2016-1, chong_many-body_2021} are 
	\begin{equation}\label{eq:mixed-semiclassical-bounds}
		\Nrm{\Dhx\sqrt{\op^\init}}{\L^p} \leq C \quad\text{ and }\quad\Nrm{\Dhv\sqrt{\op^\init}}{\L^p} \leq C.
	\end{equation} 
	In~\cite{benedikter_mean-field_2016-1}, the interaction is smooth and Equation~\eqref{eq:mixed-semiclassical-bounds} is required to hold true for the the initial state for $p=1$. In~\cite{chong_many-body_2021}, singular interactions of the form $\n{x}^{-a}$ with $a\leq 1$ are considered and Equation~\eqref{eq:mixed-semiclassical-bounds} is required to hold for $p\in [2,\infty]$.
	 
	In this work, we focus on quasi-free Gibbs states, which is an important class of examples as they model the equilibria of an ideal gas at finite positive temperature. It is well-known that the one-particle reduced density matrix associated to a Gibbs state can be computed explicitly (see Proposition~\ref{prop:quasi-free}). Using the explicit expression \eqref{eq:explicit-form}, we prove bounds on the weighted $\L^p$ norms of the quantum gradients \eqref{eq:quantum-grad} with $\op_\beta$ given by the one-particle reduced density matrix associated to the Gibbs state of a system of non-interacting fermions confined by a harmonic potential, and on its partition function. Our main result is the following.
	
	\begin{theorem}\label{thm:main-thm}
		Let $H = \frac{\n{\opp}^2+\n{x}^2}{2}$ be the $d$-dimensional harmonic oscillator Hamiltonian and
		define the family of density matrix operators given by  
		\begin{equation}\label{eq:explicit-form}
			\op_\beta : =\dfrac{1}{N\,h^d}\(\id+e^{\beta(H-\mu)}\)^{-1}
		\end{equation}
		where $\mu$ depends on $\beta$, $N$ and $h$ and is chosen so that $h^d\tr(\op_\beta)=1$ holds. Then for any $p\geq 2$,
		$\op_\beta\in\mathcal{W}^{1,p}(\sfm)$ and $\sqrt{\op_\beta}\in\mathcal{W}^{1,p}(\sfm)$ uniformly with respect to $\hbar\in(0,1)$. More explicitly, let 
		\begin{equation}\label{eq:def_Z}
			Z_\beta:= h^d\tr(e^{-\beta H}), \quad \text{ and } \quad Z_\mu:= \lambda\, e^{-\beta\mu}
		\end{equation}
		with $\lambda := Nh^d$, then the following holds. For any fixed $\beta \in \R_+$, $\hbar \in (0, 1)$
		and for $p \in [1, \infty]$, we have the bound 
		\begin{equation}\label{thm:main_result_III}
			\Nrm{\Dh\op_\beta}{\L^p} \leq C_{d,p}\, \frac{\beta^{\frac{1}{2}-\frac{d}{p}}}{Z_\mu}  \frac{\max\(2\sqrt{2},\beta\hbar\)^{\frac{1}{2}-\frac{1}{p}}}{\(\theta(\beta\hbar)\)^\frac{1}{p}}
		\end{equation}
		where $\theta(x) = \th(x)/x$ with $\th(x) = \tfrac{e^x-e^{-x}}{e^x+e^{-x}}$, and 
		\begin{align}\label{thm:main_result_I}
			C_{\lambda,\beta}^{-1}\,Z_{\beta} \leq Z_\mu \leq Z_{\beta}
		\end{align}
		with $C_{\lambda,\beta} = 2$ if $\mu \le d\hbar/2$ and $C_{\lambda, \beta}= 1 + e^{\beta\lambda^{1/d}/\pi}$ if $\mu \ge d\hbar/2$.
	\end{theorem}
	
	\begin{remark}
	As a corollary of our result, we exhibit a class of initial states satisfying the regularity assumptions of \cite[Theorem~3.1]{chong_many-body_2021}.
	\end{remark}
	
	\begin{remark}
		Our proof in the case $p=\infty$ is rather general and in particular can be adapted without difficulty to other external trapping potentials $V(x)$, that is to the case of an Hamiltonian of the form $H = \frac{\n{\opp}^2}{2} + V(x)$.
	\end{remark}
	
	\begin{remark}
		As indicated in \cite[Figure 1]{chong_many-body_2021}, there are three scaling regimes to keep in mind when studying mean-field dynamics of quantum systems (both bosonic and fermionic), namely, $N h^d\rightarrow 0, N h^d\rightarrow \infty$, and $N h^d$ converges to some constant (in fact, $N h^d=1$ which we called the critical scaling). 

		Let $\beta>0$ be fixed. In the {lower density} scaling regime $N h^d\rightarrow 0$,  we have that $e^{\beta\mu}\rightarrow 0$ or equivalently $\mu\rightarrow -\infty$, which is an immediate consequence of  \eqref{thm:main_result_I} in the above theorem and the fact that $Z_\beta$ is bounded (cf. Section \ref{sect:bounds_on_inverse_fugacity}). It is not difficult to see that in this case $\op_\beta=\(Nh^d + Z_\mu\,e^{\beta H}\)^{-1}$ will converge to $Z_\beta^{-1} \,e^{-\beta H}$. In other words, in the semiclassical limit $\hbar\rightarrow 0$, $\op_\beta$ will converge to the Maxwell--Boltzmann distribution. This result should not be surprising to the reader since in the lower density scaling regime $Nh^d\to 0$, the Fermi gas behaves like a classical ideal gas. 
		
		In the case of the higher density scaling regime $N h^d\rightarrow \infty$, we recall that as in \cite[Equation~(16)]{chong_many-body_2021}, $\lambda = Nh^d \leq \frac{1}{\Nrm{\op}{\L^\infty}} =: \cC_\infty^{-1}$. This shows that $\Nrm{\op}{\L^\infty}\rightarrow 0$ as $\lambda\rightarrow \infty$. In short, the higher density scaling regime is not accessible to Fermi gases. From a more physical perspective, a high density Fermi gas exhibits quantum ``degeneracy'' behavior, consequence of the Pauli exclusion principle, which limits the gas from getting arbitrarily dense. 
		
		Finally, in the critical scaling case, we see that $\op_\beta$ does not converge to the Maxwell--Boltzmann distribution but instead to the Fermi--Dirac distribution. The divergence from classical behavior is attributed to the fact that at the critical scaling we are in fact modeling the degeneracy behavior of the Fermi gas. See e.g.~\cite[Chapter 17.5]{chapman_mathematical_1990}.
	\end{remark}
	
	\begin{remark}
		In the semiclassical limit $\hbar\to 0$, then our main Inequality~\eqref{thm:main_result_III} implies
		\begin{equation*}
			\Nrm{\Dh\op_\beta}{\L^p} \leq C_{d,p,\lambda}\, \beta^{\frac{1}{2}+\frac{d}{p'}} + o(\hbar)
		\end{equation*}
		which is optimal in the sense that it is consistent with the behavior of the classical norm $\Nrm{\nabla\(Z_\beta^{-1}e^{-\beta\n{z}^2}\)}{L^p(\RRd)}$ in terms of $\beta$. It is however interesting to see that this scaling seems to change at low temperature and when $\hbar$ is not negligible (i.e. when $\beta\hbar$ is large).
	\end{remark}

\subsection{Second quantization formalism}
	To introduce quasi-free Gibbs states, we will need the second 
	quantized description of a many-body system of non-interacting fermions in an external trap. 
	Here we give a brief review of the Fock space formalism. For a more comprehensive treatment, the reader could consult \cite{bratteli_operator_1981, solovej_many_2014}.

	The fermionic (antisymmetric) Fock space over the one-particle Hilbert space $\h = L^2(\R^d)$ is defined to be the closure of
	\begin{equation*}
		\cF :=\bigoplus^\infty_{n=0} \h^{\wedge n} = \CC \oplus \h \oplus (\h\wedge \h)\oplus \cdots
	\end{equation*} 
	with respect to the norm induced by the standard associated inner product. Here the $n$-th sector of $\cF$ given by $\h^{\wedge n}$ denotes the subspace of $\h^n = L^2(\R^{dn})$ containing functions that satisfy the antisymmetry property
	\begin{equation*}
		\psi(x_1, \ldots, x_n) = \operatorname{sgn}(\sigma)\, \psi(x_{\sigma(1)}, \ldots, x_{\sigma(n)})
	\end{equation*} 
	for any permutation $\sigma$. An important vector to note is the vacuum vector $\Omega = (1, 0, 0, \ldots)\in \cF$ which describes a pure state with no particles. 

	The operator-valued distribution $a_x$ and $a^\ast_x$ on $\cF$ are defined by their actions on the $n^\text{th}$ sector of $\cF$ as follows. Let $\psi \in \h^{\wedge n}$, then  
	\begin{align*}
		(a_x\psi)(x_1, \ldots, x_{n-1}) &:= \sqrt{n}\,\psi(x, x_1, \ldots, x_{n-1}),
		\\
		(a_x^\ast\psi)(x_1,\ldots, x_{n+1}) &:= \sum^{n+1}_{j=1}\frac{(-1)^{j-1}}{\sqrt{n+1}}\, \delta_x(x_j)\, \psi(x_1,\ldots,\cancel{x_j},\ldots, x_{n+1}).
	\end{align*}
	Furthermore, for each $\varphi \in \h$, we associate to it the annihilation operator and its adjoint the creation operator, defined by
	\begin{equation*}
		a(\varphi)= \intd \conj{\varphi(x)}\,a_x\dd x \quad\text{ and } \quad a^\ast(\varphi)= \intd \varphi(x)\,a^\ast_x\dd x.
	\end{equation*}
	It can be shown that creation and annihilation operators are bounded on $\cF$ with operator norm $\Nrm{a(\varphi)}{\infty} = \Nrm{a^\ast(\varphi)}{\infty} = \Nrm{\varphi}{\h}$ and that they satisfy the canonical anticommutation relations, i.e. 
	\begin{equation*}
		\{a(\varphi), a(\phi)\} = \{a^\ast(\varphi), a^\ast(\phi)\} = 0 \quad \text{ and } \quad \{a^\ast(\varphi), a(\phi)\}  = \inprod{\phi}{\varphi}_\h\id
	\end{equation*}
	hold for all $(\varphi, \phi) \in \h^2$ where $\{\sfA, \sfB\}:=\sfA\sfB+\sfB\sfA$ is the standard anticommutator of the operators $\sfA$ and $\sfB$.

	In this work, we are interested in trapped non-interacting Fermionic systems modeled by the one-particle Hamilton operator 
	\begin{equation}\label{eq:trap-hamiltonian} 
		H = \tfrac{1}{2} \n{\opp}^2 + V(x),
	\end{equation}
	where $\opp = -i \hbar \grad_x$ and $V$ is an external trapping potential. The second quantization of the one-particle operator $H$ on $\cF$  is defined to be the operator $\sfH=\dG(H):=\intd a^\ast_x H a_x\dd x$ whose action on the $n^\text{th}$ sector is given by 
	\begin{equation*}
		\(\dG(H)\Psi\)(x_1,\ldots, x_n)= \sum^n_{j=1} (-\tfrac{\hbar^2}{2}\lapl_{x_j}+V(x_j))\Psi(x_1,\ldots, x_n).
	\end{equation*}
	Another useful operator on $\cF$ is the number operator
	\begin{equation*}
		\cN = \bigoplus^\infty_{n=0} n\,\id_{\h^{\wedge n}}.
	\end{equation*}
	
\subsection{Harmonic oscillators}\label{sec:harmonic_oscillator}

	To simplify the presentation, we choose the external potential $V$ to be the harmonic trapping potential, that is, 
	\begin{equation}\label{eq:harmonic-hamiltonian}
		H = \tfrac12 \n{\opp}^2+ \tfrac12 \n{x}^2 = \sum^d_{i=1} \tfrac12 \,\opp_i^2+\tfrac12\, x_i^2 = \sum^d_{i=1}H_i.
	\end{equation}
	Since $H$ is the $d$-dimensional quantum harmonic oscillator, it is well-known that it emits a discrete spectrum where the eigenvalues and eigenvectors are given by 
	\begin{align*}
		E_{n} &= \(\n{n}_1 + \tfrac{d}{2}\) \hbar &&\text{for } n = (n_1, \ldots, n_d) \in \N_{0}^d,
		\\
		u_{E_{n}} &= \phi_{n_1}\otimes\cdots \otimes\phi_{n_d} &&\text{with } H_i\phi_{n_i} = \(n_i+\tfrac12\)\hbar\, \phi_{n_i}.
	\end{align*}
	where $\n{n}_1 = n_1+\ldots+n_d$ and the functions $\phi_{n_i}$ are the standard one-variable Hermite functions
	\begin{equation*}
		\phi_{n_i}(x_i)= c_{n_i} \,e^{-x^2/2}P_{n_i}(x_i) \quad \text{ with } \quad c_{n_i}^{-1} = 2^{\frac{n_i}{2}} (n_i!)^{\frac12} \(\pi\hbar\)^\frac14, 
	\end{equation*}
	with $P_{n_i}(x)$ the Hermite polynomials. Moreover, for each eigenvalue $E_{n}$ the corresponding multiplicity is given by $g_{\n{n}_1, d} = \binom{\n{n}_1+d-1}{d-1}$. In particular, by the spectral decomposition of $H$, we could now rewrite $\sfH$ as 
	\begin{equation}
		\sfH = \sum_{n\in \Z_0^d} E_{n} \, a^\ast(u_{E_n})\,a(u_{E_n}).		
	\end{equation}

\subsection{Thermal states}
	
	We consider quantum states that are given by density operators $\opom$ on $\cF$, that is, the expectation of an observable $\sfA$ is given by $\bangle{\sfA}_{\opom} = \tr_{\cF}\(\sfA\, \opom\)$ for every $\sfA \in \cB(\cF)$, where $\cB(\cF)$ denotes the space of bounded operators on the Fock space. The Gibbs equilibrium of a trapped non-interacting Fermi gas associated to some positive temperature $T>0$ is defined to be the unique minimizer of the Gibbs free energy functional
	\begin{equation*}
		\ecF(\opom) = \tr_{\cF}\(\sfH\,\opom\) - \kb\, T\,\ecS(\opom) \quad \text{ with } \quad \ecS(\opom)=-\tr_{\cF}\(\opom\, \log\opom\)
	\end{equation*}
	in $\cX_N=\{\opom \in \cB(\cF) \mid \opom \ge 0,\ \tr_{\cF}\(\opom\)=1,\ \tr_{\cF}\(\cN\,\opom\)=N\}$ where $\ecS$ is the von Neumann entropy. It can be checked that the Gibbs state associated to the temperature $T$ is given by the normalized positive trace class operator
	\begin{equation}\label{def:gibbs_states}
		\gibbs = \frac{1}{\cZ_N}\,e^{-\beta \(\sfH - \mu\cN\)}
	\end{equation}
	where $\beta = 1/(\kb\,T)$ and the chemical potential $\mu$ is chosen so that $\Tr{\cN \gibbs}=N$. 
	Here $\cZ_N$ is the grand canonical partition function
	\begin{equation*}
		\cZ_N = 1+\sum_{n=1}^\infty e^{n\beta\mu} \tr_{\h^{\wedge n}} \exp\bigg(-\beta \sum^{dn}_{j=1} H_j\bigg)
		= \prod_{n \in \Z_0^d} (1+e^{-\beta (E_{n}-\mu)}).
	\end{equation*} 
	
\subsection{Quasi-free states}

	Here we will only give a rudimentary introduction to the tools necessary for the subsequent sections. A more comprehensive exposition of the following content can be found in \cite{bratteli_operator_1981, solovej_many_2014}. The state $\opom$ is said to be quasi-free if for all $n\in \N$, we have that
	\begin{align*}
		&\small \tr_{\cF}\(a^{\sharp_1}(f_1)\cdots a^{\sharp_{2n+1}}(f_{2n+1})\opom\) = 0,
		\\
		&\small \tr_{\cF}\(a^{\sharp_1}(f_1)\cdots a^{\sharp_{2n}}(f_{2n})\opom\) = \sum_{\sigma} (-1)^{\sigma} \prod^n_{j=1} \Tr{a^{\sharp_{\sigma(j)}}(f_{\sigma(j)})a^{\sharp_{\sigma(j+n)}}(f_{\sigma(j+n)})\opom},
	\end{align*}
	where $a^{\sharp_j}$ stands for either $a$ or $a^*$, and the sum is taken over all permutations $\sigma$ satisfying $\sigma(1) < \sigma(2)<\ldots <\sigma(n)$, and $\sigma(j) < \sigma(j+n)$, for all $j\in\set{1,\ldots, n}$. The definition indicates that quasi-free states are determined by the one-particle reduced density matrix operator defined via its integral kernel 
	\begin{align}
		\op(x, y)= \frac{1}{N h^d}\tr_{\cF}\(a^\ast_y\, a_x\, \opom\),
	\end{align}
	and the antisymmetric pairing function defined by 
	\begin{align}
		\opal(x, y) = \frac{1}{N h^d}\tr_{\cF}\(a_y \, a_x\, \opom\).
	\end{align}
	More compactly, we introduce the generalized one-particle density matrix $\opG:\h\oplus\h^\ast\rightarrow \h\oplus\h^\ast$ defined by 
	\begin{align}\label{def:generalized_opdm}
		\opG = 
		\begin{bmatrix}
			\op & \opal\\
			-\mathds{J}\opal\mathds{J} & \id -\mathds{J}\op \mathds{J}^\ast
		\end{bmatrix}
	\end{align}
	where $\mathds{J}:\h\rightarrow \h^\ast$ is  the map $\mathds{J}(\phi)=\inprod{\phi}{\cdot}$ and $\mathds{J}^\ast$ is its adjoint operator. This observation is summarized by the following proposition (see e.g. \cite[Appendix G]{solovej_many_2014}).
	
	\begin{prop}
		Let $\opG:\h\oplus\h^\ast\rightarrow \h\oplus\h^\ast$ be an operator of the form \eqref{def:generalized_opdm}. Then $\opG$ is the generalized one-particle density matrix of some quasi-free state with finite particle number if and only if $\opG\ge 0$ and $\tr\op<\infty$. 
	\end{prop}

	In the case of $\gibbs$, it can be readily shown that its generalized one-particle density matrix is given by 
	\begin{equation*}
		\opG_{\beta}
		=
		\begin{bmatrix}
			\op_\beta & 0
			\\
			0 & \id-\mathds{J}\op_\beta \mathds{J}^\ast
		\end{bmatrix}
	\end{equation*}
	where $\opal = 0$, which we refer to as the gauge-invariant condition. The following proposition, whose proof can be found for example in \cite[Proposition~5.2.23]{bratteli_operator_1981}, makes the link between the Fock space Gibbs state~\eqref{def:gibbs_states} and the one body Gibbs state considered in our main theorem.
	
	\begin{prop}\label{prop:quasi-free}
		Let $\gibbs$ denotes the Gibbs state \eqref{def:gibbs_states} with $\Tr{e^{-\beta\sfH}} <\infty$. Then $\gibbs$ is quasi-free and its one-particle reduced density matrix operator is given by Equation~\eqref{eq:explicit-form}.
	\end{prop}

\section{Semiclassical thermal regularity}

In this section, we give the proof of the main theorem. 

\subsection{Preliminaries}	
	For a trapped non-interacting Fermi gas with Hamiltonian~\eqref{eq:trap-hamiltonian}, let us write $G = e^{-\beta H}$ and $\lambda = N h^d$. Consider the quantities
	\begin{subequations}
		\begin{align}
			\opg_\beta &= Z_\beta^{-1} \,e^{-\beta H},
			\\
			\op_\beta &= \lambda^{-1}\(\id + e^{\beta (H-\mu)}\)^{-1} = G_\mu \(\id+\lambda\, G_\mu\)^{-1},
		\end{align}
	\end{subequations}
	where $G_\mu = \lambda^{-1}\,e^{-\beta\(H-\mu\)} = Z_\mu^{-1} \,G$ and $Z_\beta$ and $Z_\mu$ are given in Equation~\eqref{eq:def_Z}. The quantum gradients of these operators are given by the following lemma.
	\begin{lem}
		Let $\Dh$ be a quantum gradient. Then we have
		\begin{subequations}
			\begin{align}\label{eq:gradient_gaussian}
				\Dh\opg_\beta &= -\frac{\beta}{Z_\beta} \int_0^1 G^{1-s} \(\Dh H\) G^s \dd s,
				\\\label{eq:gradient_fermi_dirac}
				\Dh\op_\beta &= -\beta \int_0^1 G_\mu^{1-s}(\id + \lambda G_\mu)^{-1} \(\Dh H\) (\id + \lambda G_\mu)^{-1} G_\mu^s\dd s.
			\end{align}
		\end{subequations}
	\end{lem}

	\begin{proof}
		Notice that for any operators $A$ and $B$, we have that 
		\begin{equation*}
			\dpt \com{A,e^{tB}} = \com{A,B\,e^{tB}} = \com{A,B}e^{tB} + B\, \com{A,e^{tB}}.
		\end{equation*}
		Therefore, $\dpt\(e^{-tB}\com{A,e^{tB}}\) = e^{-tB}\com{A,B}e^{tB}$ and we obtain the Duhamel-like formula
		\begin{equation}
			\com{A,e^{tB}} = \int_0^t e^{\(t-s\)B}\com{A,B}e^{sB}\dd s.
		\end{equation}
		In particular, taking $t=1$ and $B=-\beta H$, we obtain the identity
		\begin{equation*}
			\Dh G = -\beta \int_0^1 G^{1-s} \(\Dh H\) G^s\dd s
		\end{equation*}
		from which we deduce Identity~\eqref{eq:gradient_gaussian}. Then observe that if $A$ is an invertible operator, by the Leibniz rule for commutators, $0 = \Dh\!\(AA^{-1}\) = \(\Dh A\)A^{-1} + A\Dh\!\(A^{-1}\)$ and so
		\begin{equation}\label{eq:derivative_A_inverse}
			\Dh\!\(A^{-1}\) = - A^{-1} \(\Dh A\)A^{-1}.
		\end{equation}
		In particular, since $\Dh\(\id+\lambda G_\mu\) = \lambda \Dh G_\mu$, we deduce that
		\begin{align*}
			\Dh\op_\beta &= \(\Dh G_\mu\) \(\id + \lambda G_\mu\)^{-1} - G_\mu \(\id+\lambda G_\mu\)^{-1} \lambda \(\Dh G_\mu\) \(\id + \lambda G_\mu\)^{-1}
			\\
			&= \(\id + \lambda G_\mu\)^{-1} \(\Dh G_\mu\) \(\id + \lambda G_\mu\)^{-1}
		\end{align*}
		which leads to Formula~\eqref{eq:gradient_fermi_dirac}.
	\end{proof}
	Taking the square root of $G$ changes $\beta$ by $\beta/2$. Hence it follows that
	\begin{equation*}
		\Dh\sqrt{\opg_\beta} = -\frac{\beta}{2\sqrt{Z_\beta}} \int_0^1 G^{(1-s)/2} \(\Dh H\) G^{s/2} \dd s.
	\end{equation*}
	On the other hand, it is more difficult to compute explicitly $\Dh\sqrt{\op_\beta}$. The following lemma allows us to bound $\sqrt{\op_{\beta}}$ by reducing the problem to estimating $G$.
	
	\begin{lem}
		Let $\sfm$ be a self-adjoint operator. Then
		\begin{equation*}
			\Nrm{\Dh\sqrt{\op_\beta}}{\L^p(\sfm)} \leq \(\Nrm{\Dh \sqrt{G_\mu}}{\L^p(\sfm)} + \frac{\lambda}{2} \Nrm{\sqrt{\op_\beta}}{\L^q(\sfm)} \Nrm{\Dh G_\mu}{\L^r}\)
		\end{equation*}
		for any $1\le p, q, r\le \infty$ such that $\frac{1}{p}= \frac{1}{q} + \frac{1}{r}$.
	\end{lem}
	
	\begin{proof}
		By Identity \eqref{eq:derivative_A_inverse}, we have that 
		\begin{equation}\label{eq:gradient_sqrt}
		\begin{aligned}
			\Dh\sqrt{\op_\beta}
			= \(\Dh \sqrt{G_\mu} - \sqrt{\op_\beta} \(\Dh \sqrt{\id+\lambda G_\mu}\)\) \(\id+\lambda G_\mu\)^{-\frac{1}{2}}.
		\end{aligned}
		\end{equation}
		To bound $\Dh \sqrt{\id+\lambda G_\mu}$, we proceed as in \cite[Lemma~7.1]{chong_many-body_2021}. Let $A := \lambda G_\mu$.
		Since $0\leq A \leq 2\,c := \lambda Z_\mu^{-1}$, we deduce that $\Nrm{A-c}{\L^\infty} \leq c$ and so the following series is absolutely convergent 
		\begin{equation*}
			\sqrt{\id + A} = \sqrt{1+c}\,\sqrt{\id+\tfrac{1}{1+c}\(A-c\)} = \sum_{n=0}^\infty \binom{1/2}{n} \(\tfrac{1}{1+c}\)^{n-\frac{1}{2}} \(A-c\)^n.
		\end{equation*}
		Since it follows from the Jacobi identity that 
		\begin{equation*}
			\Dh\(A-c\)^n = \sum_{k=1}^n \(A-c\)^{k-1} \(\Dh A\) \(A-c\)^{n-k},
		\end{equation*}
		then we deduce the estimate
		\begin{align*}
			\Nrm{\Dh\sqrt{\id+A}}{\L^p} &\leq \sum_{n=1}^\infty \n{\binom{1/2}{n}} \(\tfrac{1}{c+1}\)^{n-\frac{1}{2}} n \Nrm{\Dh A}{\L^p} \Nrm{A-c}{\L^\infty}^{n-1}
			\\
			&\leq \frac{1}{2\sqrt{c+1}} \,\sum_{n=1}^\infty \binom{-1/2}{n-1} \(\tfrac{-c}{c+1}\)^{n-1} \Nrm{\Dh A}{\L^p} = \frac{1}{2} \Nrm{\Dh A}{\L^p}.
		\end{align*}
		We conclude by applying the fact that $\Nrm{\Dh\sqrt{\op_\beta} \,\sfm}{\L^p}= \Nrm{\sfm\,\Dh\sqrt{\op_\beta}}{\L^p}$, which follows from taking the adjoint, and the triangle inequality to Formula~\eqref{eq:gradient_sqrt}. This yields the desired result.
	\end{proof}
		
\subsection{Bounds on the inverse fugacity}\label{sect:bounds_on_inverse_fugacity}

	In the rest of this paper, we assume $H$ is the harmonic oscillator Hamiltonian $H = \frac{\n{\opp}^2 + \n{x}^2}{2}$. In this case of the partition function in $\opg_\beta$ has the closed form
	\begin{equation}
		Z_\beta =  \(\frac{2\pi}{\beta}\)^d \frac{1}{\operatorname{shc}(\frac{\beta\hbar}{2})^d}
	\end{equation}
	where $\operatorname{shc}(x) = \sh(x)/x$ denotes the hyperbolic sinc function. In particular $Z_\beta \leq (2\pi/\beta)^d$ and $Z_\beta \sim (2\pi/\beta)^d$ as $\beta\hbar\rightarrow 0$. The function $Z_\mu$, that one could call the inverse fugacity, can be compared to the partition function $Z_\beta$ as proved in the following proposition.

    \begin{prop}\label{prop:bounds_on_chemical_potential}
		Let $Z_\mu= \lambda\, e^{-\beta\mu}$ with $\lambda = Nh^d$. Then the following inequality holds
		\begin{equation}\label{eq:bound_cN}
			C_{\lambda,\beta}^{-1}\,Z_{\beta} \leq Z_\mu \leq Z_{\beta}
		\end{equation}
		with $C_{\lambda,\beta} = 2$ if $\mu \leq d\hbar/2$ and $C_{\lambda,\beta} = 1 + e^{\beta\lambda^{1/d}/2\pi}$ if $\mu \ge d\hbar/2$.
	\end{prop}
	
	\begin{proof}[Proof of Proposition \ref{prop:bounds_on_chemical_potential}]
		Let $c_\mu = e^{\beta\mu} = N h^d/Z_\mu$. Notice that
		\begin{equation*}
			1 = h^d \Tr{\op_\beta} = \frac{1}{N} \Tr{\big(\id + c_\mu^{-1}e^{\beta H}\big)^{-1}} \leq \frac{c_\mu}{N} \Tr{e^{-\beta H}},
		\end{equation*}
		which implies the first inequality in Formula~\eqref{eq:bound_cN}.
		
		Next, define the function $g:\R_+\to\R_+$ by $g(r) = \big(1+c_\mu^{-1}\,e^{\beta r}\big)^{-1}$, then for any $R>0$, it holds $g(r) \geq \big(1+c_\mu^{-1}\,e^{\beta r}\big)^{-1} \,\indic_{r\leq R}$, which leads to
		\begin{equation*}
			1 = \frac{1}{N} \tr\(g(H)\) \geq \frac{1}{N\(1+c_\mu^{-1}\,e^{\beta R}\)}\, \tr\(\indic_{H\leq R}\).
		\end{equation*}
		Recalling the property of the harmonic oscillator given in Section~\ref{sec:harmonic_oscillator}, one sees that the trace of the characteristic function is nothing but
		\begin{equation*}
			\Tr{\indic_{H\leq R}} = \n{\Set{n\in\N^d_0 : \(\n{n}_1+\tfrac{d}{2}\)\hbar \leq R}}.
		\end{equation*}
		Since $\n{n}_1 \geq \sup_{j} n_j =: \n{n}_\infty$, then this can be crudely estimated by 
		\begin{equation}
			\Tr{\indic_{H\leq R}} \geq \n{\Set{n\in\N^d_0 : \n{n}_\infty \leq \tfrac{R}{\hbar}-\tfrac{d}{2}}} = {\(\floor{\tfrac{R}{\hbar}-\tfrac{d}{2}}+1\)^d}.
		\end{equation}
		In particular, since $\mu\ge d\hbar/2>0$, then taking $R = \mu = \frac{\ln c_\mu}{\beta}$ yields $\tfrac{\mu}{\hbar}-\frac{d}{2} \leq \big(N^{1/d}+1\big)$ and so since $N \geq 1$,
		\begin{equation}\label{eq:bound_mu}
			\mu \leq 2\,N^{1/d}\hbar + \frac{d\hbar}{2}.
		\end{equation}

		Now, let us obtain upper bounds for $c_\mu$. First, observe that we have the lower bounds  
		\begin{equation}\label{ineq:lower_for_c_N}
			1 = \frac{c_\mu}{N} \Tr{\big(\id + c_\mu e^{-\beta H}\big)^{-1}e^{-\beta H}} \geq \frac{Z_\beta\,c_\mu}{Nh^d\(1+c_\mu e^{-d\beta\hbar/2}\)}.
		\end{equation}
		If $c_\mu \le e^{d\beta\hbar/2}$ (i.e. $\mu\le d\hbar/2$) then it follows from Inequality \eqref{ineq:lower_for_c_N} that 
		\begin{align*}
			c_\mu \le 2\, N h^d Z_\beta^{-1}.
		\end{align*}
		On the other hand, if $c_\mu \ge e^{d\beta\hbar/2}>1$ (i.e. $\mu \ge d\hbar/2$), then it follows from Inequality \eqref{eq:bound_mu} and Inequality \eqref{ineq:lower_for_c_N} that we have the bound
		\begin{align}
				c_\mu \leq \(1+c_\mu\, e^{-d\beta\hbar/2}\) N h^d Z_\beta^{-1} \leq \(1+e^{2\beta\hbar N^{1/d}}\) N h^d Z_\beta^{-1} 
		\end{align}
		This completes our proof of the proposition.
	\end{proof}

	\subsection{\texorpdfstring{$\L^\infty$}{L-infty} bounds} This section is devoted to the proof of the following proposition.
	\begin{prop}\label{prop:regu_thermal_harmonic}
		Let $\beta >0$ and $H = \frac{\n{\opp}^2 + \n{x}^2}{2}$, then we have the estimates
		\begin{align*}
			\Nrm{\Dh\opg_\beta}{\L^\infty} &\leq \frac{2}{Z_\beta}\max\(\sqrt{\beta}, \beta\sqrt{\hbar}\),
			\\
			\Nrm{\Dh\op_\beta}{\L^\infty} &\leq \frac{2}{Z_\mu} \max\(\sqrt{\beta}, \beta\sqrt{\hbar}\).
		\end{align*}
	\end{prop}

	Prior to giving the proof of the above proposition, let us make the following observation. Since $\Dhx{H} = x$ and $\Dhv{H} = \opp$, we deduce that $\Dh_{z} H =\opz := (x,\opp)$ where $z=(x, \xi)$. By Identity \eqref{eq:gradient_gaussian}, we have that
	\begin{equation*}
		\Dh\op_\beta = -\beta \int_0^1 (\id + \lambda G_\mu)^{-1}G_\mu^{1-s} \,\opz\, G_\mu^s(\id + \lambda G_\mu)^{-1}\dd s
	\end{equation*}
	which implies the estimate 
	\begin{equation*}
		\Nrm{\Dh\op_\beta}{\L^p} \leq \beta\,\int_0^1 \Nrm{G_\mu^{1-s} \opz\, G_\mu^s}{\L^p} \d s = \beta\, Z_{\mu}^{-1}\,\int_0^1 \Nrm{G^{1-s} \opz\, G^s}{\L^p} \d s.
	\end{equation*}
	Since $\Nrm{G}{\L^\infty}\leq 1$, then we have the estimate
	\begin{equation}\label{eq:split_integral}
		\begin{aligned}
			\Nrm{\Dh\op_\beta}{\L^p} &\leq \beta\, Z_\mu^{-1}\int_0^{1/2}  \Nrm{G^{1-s} \opz}{\L^p} \d s + \beta\, Z_\mu^{-1}\int_{1/2}^{1}  \Nrm{\opz \,G^{s}}{\L^p} \d s
			\\
			&\leq 2\,\beta\, Z_\mu^{-1} \int_{1/2}^1  \Nrm{\opz\,G^{s}}{\L^p} \d s.
		\end{aligned}
	\end{equation}
    Hence to estimate $\Dh\op_\beta$ and $\Dh\opg_\beta$, it remains to estimate the value of $\Nrm{x\,e^{-\beta s H}}{\L^p}$ and $\Nrm{\opp\,e^{-\beta s H}}{\L^p}$ for $s\in[\frac12,1]$. Let start with the case $p=\infty$.
	\begin{lem}\label{lem:weights_vs_gaussian}
		Let $\beta>0$, then we have the estimate
		\begin{align*}
			\Nrm{\n{x}^n e^{-\beta H}}{\L^\infty}^{2/n} = \Nrm{\n{\opp}^n e^{-\beta H}}{\L^\infty}^{2/n} \leq n\max\(\tfrac{2}{\beta}, \sqrt{2}\,\hbar\).
		\end{align*}
	\end{lem}
	
	\begin{remark}
		In the classical case, it is not difficult to prove that the maximum of the function $x\mapsto \n{x}^n e^{-\beta \n{x}^2}$ is $\(\frac{n}{2e\beta}\)^{n/2}$, and more generally,
		\begin{equation*}
			\Nrm{\n{x}^n e^{-\beta \n{x}^2}}{L^p}^p = \omega_d\,\Gamma(\tfrac{d+np}{2}) \(\beta p\)^{-\frac{d+np}{2}}
		\end{equation*}
		where $\omega_d$ is the volume of the $d$-dimensional unit ball and $\Gamma$ is the gamma function.
	\end{remark}
	
	\begin{proof}[Proof of Lemma~\ref{lem:weights_vs_gaussian}]
		It is sufficient to prove a bound on the first quantity because $H$ is symmetric in $x$ and $\opp$. Let $\varphi\in L^2$ and $\psi = e^{-t H} \varphi$. Notice that for any $t\geq 0$, $\Nrm{\psi}{L^2} \leq \Nrm{\varphi}{L^2}$. Let $y := \Nrm{\n{x}^{n} \psi}{L^2}^{2/n}$. Since $2\,\dpt\psi = - \big(\n{x}^2-\hbar^2\Delta\big)\psi$, then integrating by parts yields
		\begin{multline*}
			\dt y^n = - \intd\n{\psi}^2 \n{x}^{2\(n+1\)} \d x - \hbar^2 \Re{\intd \nabla\(\conj{\psi} \n{x}^{2n}\) \cdot\nabla\psi\dd x}
			\\
			\leq - \intd \n{\psi}^2 \n{x}^{2\(n+1\)} + \hbar^2 \n{\nabla\psi}^2 \n{x}^{2n} \d x + 2\,n\,\hbar^2  \intd \n{\psi} \n{x}^{2n-1} \n{\nabla\psi}\dd x.
		\end{multline*}
		Applying Young's inequality for the product, we get
		\begin{align*}
			\dt y^n &\leq - \intd \n{\psi}^2 \n{x}^{2(n+1)} \d x + \(n\hbar\)^2 \intd\n{\psi}^2 \n{x}^{2\(n-1\)}\d x
			\\
			&\leq - c^{-1}\,y^{n+1} + \(n\hbar\)^2  c\, y^{n-1}
		\end{align*}
		where $c = \Nrm{\varphi}{L^2}^{2/n}$. This yields the differential inequality
		\begin{equation}
			 y' \leq  -\frac{1}{n\,c}\(y^{2} - \(n\hbar\,c\)^2\).
		\end{equation}
		This ordinary differential equation has a fixed point at $y = n\hbar\,c$. If initially, $y \leq \hbar\,c$, then $y'\geq 0$ but $y$ remains smaller than $n\hbar\, c$. If not, then at any time $y'< 0$ and $\hbar\,c < y(t) < y(0)$. If initially $y > \sqrt{2}\,n\,\hbar \,c$, then as long as it remains true, it holds
		\begin{equation*}
			y' \leq -\frac{1}{2\,n\,c}\, y^2
		\end{equation*}
		which implies
		\begin{equation*}
			y(t) \leq \max\(\frac{2}{t}, \sqrt{2}\,\hbar\) n\, c
		\end{equation*}
		and proves the result by taking $t=\beta$.
	\end{proof}
	
	To complete the proof of Proposition~\ref{prop:regu_thermal_harmonic}, we use the above lemma to get that for any $s\in[\tfrac12,1]$,
	\begin{equation*}
		\Nrm{x\,e^{-\beta s H}}{\L^\infty} \leq 2\max\(\tfrac{1}{\sqrt{\beta}}, \sqrt{\hbar}\).
	\end{equation*}
	Finally, we conclude using Inequality~\eqref{eq:split_integral}.
	
	\subsection{\texorpdfstring{$\L^p$}{Lp} bounds for \texorpdfstring{$2\le p<\infty$}{2<= p<infty}}
	\begin{prop}\label{prop:Lp_regu_thermal_harmonic}
		Let $\beta> 0$ and $\hbar \in (0,1)$. Suppose $p \in [2, \infty]$ then there exists $C_{d,p} > 0$ such that
		\begin{align*}
			\Nrm{\Dh\opg_\beta}{\L^p} & C_{d,p}\, \frac{\beta^{\frac{1}{2}-\frac{d}{p}}}{Z_\beta}  \frac{\max\(2\sqrt{2},\beta\hbar\)^{\frac{1}{2}-\frac{1}{p}}}{\theta(\beta\hbar)^\frac{1}{p}},
			\\
			\Nrm{\Dh\op_\beta}{\L^p} &\leq C_{d,p}\, \frac{\beta^{\frac{1}{2}-\frac{d}{p}}}{Z_\mu}  \frac{\max\(2\sqrt{2},\beta\hbar\)^{\frac{1}{2}-\frac{1}{p}}}{\theta(\beta\hbar)^\frac{1}{p}}.
		\end{align*}
		where where $\theta(x) = \th(x)/x$ with $\th(x) = \tfrac{e^x-e^{-x}}{e^x+e^{-x}}$.
	\end{prop}

	\begin{lem}[Moments of the thermal state]\label{lem:moments_of_thermal_state}
		Let $\beta>0$ and $n>-d$. Then
		\begin{align*}
			h^d\Tr{\n{x}^n\opg_\beta} = h^d\Tr{\n{\opp}^n\opg_\beta} = \frac{C_{d,n}}{\(\beta\,\theta(\beta\hbar/2)/2\)^{n/2}}
		\end{align*}
		where $C_{d,n} = \Gamma(\tfrac{d+n}{2})/\Gamma(\tfrac{d}{2})$.
	\end{lem}
	
	\begin{proof}
		Let $H_\circ = \frac{{x}^2+\n{\opp}^2}{2}$ be the one-dimensional Hamiltonian of the harmonic oscillator, and $\psi_n$ be its eigenvalues verifying
		\begin{equation*}
			H_\circ\psi_n = \(n+\tfrac{1}{2}\)\hbar\,\psi_n.
		\end{equation*}
		The Wigner transform of the corresponding density operator $\ket{\psi_n}\bra{\psi_n}$, is classically given by (see e.g. \cite[Section~5.04]{groenewold_principles_1946} or \cite[Theorem~1.105]{folland_harmonic_1989})
		\begin{equation*}
			f_n(z) = 2\(-1\)^n e^{-\n{z}^2/\hbar}\, L_n\!\(\frac{2\n{z}^2}{\hbar}\)
		\end{equation*}
		where $z = (x,\xi)$ and $L_n$ is the Laguerre polynomial of order $n$ defined by
		\begin{equation*}
			L_n(z) = \frac{e^{x}}{n!} \,\partial_x^n\!\(x^ne^{-x}\) = \sum_{k=0}^{n} \binom{n}{k} \frac{(-x)^k}{k!}
		\end{equation*}
		By the formula of the generating function of the Laguerre polynomials, we deduce that
		\begin{equation*}
			\sum_{n=0}^\infty t^n f_n(z) = \frac{2}{1+t}\, e^{-\frac{\n{z}^2}{\hbar}\frac{1-t}{1+t}}
		\end{equation*}
		Taking $t = e^{-\beta\hbar}$, we obtain the Wigner transform of the thermal state. Since in dimension $d$ it is factorized, it yields to the following formula for the Wigner transform of $Z_\beta^{-1}\,e^{-\beta H}$
		\begin{equation*}
			f_\beta(z) = \(\frac{\beta\,\theta(\beta\hbar/2)}{2\pi}\)^d  e^{-\beta \n{z}^2\theta(\beta\hbar/2)}
		\end{equation*}
		Its spatial moments are given by
		\begin{align*}
			h^d\Tr{\n{x}^n\opg_\beta} &= \iintd f_\beta(z)\n{x}^n\d z = \(\frac{\beta\,\theta(\beta\hbar/2)}{2\pi}\)^{d/2} \intd \n{x}^n e^{-\beta \n{x}^2\theta(\beta\hbar/2)}\dd x
		\end{align*}
		which yields the result.
	\end{proof}
	
	\begin{proof}[Proof of Proposition \ref{prop:Lp_regu_thermal_harmonic}]
		Since $h^d\Tr{\n{\n{x}^n G}^2} = Z_{2\beta}\,h^d\Tr{\n{x}^{2n} \opg_{2\beta}}$, then, by Lemma \ref{lem:moments_of_thermal_state}, we have the identity
		\begin{equation*}
			\Nrm{\n{x}^n G}{\L^2}^2 = \Nrm{\n{\opp}^n G}{\L^2}^2 = \frac{C_{d,2n}\,Z_{2\beta}}{\(\beta\,\theta(\beta\hbar)\)^{n}}
		\end{equation*}
		whenever $n>-d/2$. Now applying linear interpolation of Schatten norms and Lemma~\ref{lem:weights_vs_gaussian}, we obtain the intermediate Schatten norm bounds
		\begin{equation*}
			\Nrm{\n{x}^n G}{\L^p} = \Nrm{\n{\opp}^n G}{\L^p} \leq C_{d,n,p} \, Z_{2\beta}^\frac{1}{p} \, \frac{\max\(\tfrac{2}{\beta}, \sqrt{2}\,\hbar\)^{n\(\frac{1}{2}-\frac{1}{p}\)}}{\(\beta\,\theta(\beta\hbar)\)^\frac{n}{p}}
		\end{equation*}
		where $C_{d,n,p} = \, C_{d,2n}^\frac{1}{p} \, n^{n\(\frac{1}{2}-\frac{1}{p}\)}$. Finally, by Identity \eqref{eq:gradient_gaussian}, we arrive at 
		\begin{multline*}
			\Nrm{\Dh \opg_\beta}{\L^p} 
			\leq \frac{2\beta}{Z_\beta}\int_{1/2}^1  \Nrm{\opz\,G^{s}}{\L^p} \d s
			\\
			\leq C_{d,1,p} \frac{2\beta}{Z_\beta} \sup_{s\in[\frac12,1]} \, Z_{2s\beta}^\frac{1}{p} \, \frac{\max\(\tfrac{2}{s\beta}, \sqrt{2}\,\hbar\)^{\frac{1}{2}-\frac{1}{p}}}{\(s\beta\,\theta(s\beta\hbar)\)^\frac{1}{p}}
			\\
			\leq C_{d,p}\, \frac{\beta^{\frac{1}{2}-\frac{d}{p}}}{Z_\beta}  \frac{\max\(2\sqrt{2},\beta\hbar\)^{\frac{1}{2}-\frac{1}{p}}}{\(\theta(\beta\hbar)\)^\frac{1}{p}}
		\end{multline*}
		where we used the fact that $Z_\beta \leq (2\pi/\beta)^d$ and $C_{d,p} = 2^{\frac{5}{4}+\frac{2d+1}{p}}\,C_{d,1,p}\,\pi^\frac{d}{p}$. Similarly, Inequality~\eqref{eq:split_integral} implies the same bound for $\op_\beta$ with $Z_\beta$ is replaced by $Z_\mu$.
	\end{proof}

{\bf Acknowledgements.} J.C. was supported by the NSF through the RTG grant DMS- RTG 184031. C.S. acknowledges the support of the Swiss National Science Foundation through the Eccellenza project PCEFP2\_181153 and of the NCCR SwissMAP.


\bibliographystyle{abbrv} 
\bibliography{Vlasov}

\end{document}